\documentclass[aps,pra,twocolumn,superscriptaddress,amsfonts,floatfix,showpacs]{revtex4-1}
\usepackage{amssymb,amstext,amsmath}
\usepackage{amsthm}
\usepackage{mathrsfs}
\usepackage{mathtools}
\usepackage{graphicx}
\usepackage{float}
\usepackage[english]{babel}
\newtheorem{example}{Protocol}
\newtheorem{lemma}{Lemma}
\newtheorem{definition}{Definition}
\usepackage{hyperref}
\usepackage[dvipsnames]{xcolor}

\usepackage{tasks}
\settasks{
	counter-format=(tsk[r]),
	label-width=4ex
}

\usepackage{thm-restate}

\newcommand{\ket}[1]{| #1 \rangle}

\newcommand{\braket}[1]{\langle #1 \rangle}
\newcommand{\poly}{\operatorname{poly}}

\begin{document}

\title{Quantum sketching protocols for Hamming distance and beyond}
	
\author{Jo\~{a}o F. Doriguello}
\email{joao.doriguellodiniz@bristol.ac.uk}
\affiliation{School of Mathematics, University of Bristol, Bristol BS8 1TW, United Kingdom}
\affiliation{Quantum Engineering Centre for Doctoral Training, University of Bristol, Bristol BS8 1TW, United Kingdom}

\author{Ashley Montanaro}
\email{ashley.montanaro@bristol.ac.uk}
\affiliation{School of Mathematics, University of Bristol, Bristol BS8 1TW, United Kingdom}


	\begin{abstract}
		In this work we use the concept of quantum fingerprinting to develop a quantum communication protocol in the simultaneous message passing model that calculates the Hamming distance between two $n$-bit strings up to relative error $\epsilon$. The number of qubits communicated by the protocol is polynomial in $\log{n}$ and $1/\epsilon$, while any classical protocol must communicate $\Omega(\sqrt{n})$ bits.
        Motivated by the relationship between Hamming distance and vertex distance in hypercubes, we apply the protocol to approximately calculate distances between vertices in graphs that can be embedded into a hypercube such that all distances are preserved up to a constant factor. Such graphs are known as $\ell_1$-graphs. This class includes all trees, median graphs, Johnson graphs and Hamming graphs. Our protocol is efficient for $\ell_1$-graphs with low diameter, and we show that its dependence on the diameter is essentially optimal. Finally, we show that our protocol can be used to approximately compute $\ell_1$-distances between vectors efficiently.
	\end{abstract}
	
	\maketitle


	\section{Introduction}
    
    Imagine that two separated parties (Alice and Bob) each have some data, and would like to determine how alike their data is, using the minimal amount of communication possible. Also imagine that they are not allowed to communicate with each other, but are each only allowed to send a single message to a third party (``referee''), and do not share any prior information with each other. This communication model is known as the simultaneous message passing (SMP) model with private randomness~\cite{buhrman10}. It encapsulates, for example, a scenario where it is not clear in advance whose data sets are to be compared. Another motivation comes from cryptographic scenarios. For example, it could be that the inputs to the two parties are controlled by an adversary, who has access to any previously shared randomness and can choose the inputs such that the protocol fails~\cite{mironov11}; alternatively, Alice and Bob may simply want to find an efficient protocol which hides their data from the referee.
    
    A natural strategy for completing this task is for each of Alice and Bob to compress their data to some kind of ``sketch''~\cite{feigenbaum06,baryossef04b}, and send the sketches to the referee, who uses them to determine the distance between the corresponding original data sets. Unfortunately, even for one of the simplest distance measures possible -- testing equality of $n$-bit strings -- and even if Alice and Bob are allowed a small probability of failure, this task requires $\Theta(\sqrt{n})$ bits of classical communication~\cite{Ambainis1996,MR1427554}. In comparison, if Alice and Bob are allowed access to a shared random bit-string, this complexity drops to $O(1)$~\cite{612319}.
    
    Remarkably, the use of quantum information allows an exponential reduction in the complexity of equality-testing. If Alice and Bob encode their $n$-bit strings as particular quantum states called \emph{quantum fingerprints}, then there exists a quantum protocol that communicates only $O(\log{n})$ qubits~\cite{PhysRevLett.87.167902} and succeeds with arbitrarily high constant probability.
    
    This surprising result sparked significant interest from the perspective of computer science~\cite{Yao:2003:PQF:780542.780554,gavinsky2006strengths} and information theory~\cite{winter04}, as well as physics. Theoretically, it has been used to shed new light on the two-slit experiment~\cite{massar05} and detailed studies of fingerprinting schemes using few qubits have been undertaken~\cite{beaudrap04,scott07}. Proof-of-principle quantum fingerprinting experiments have been carried out with states of 1 qubit realized using linear optics~\cite{horn05} and nuclear magnetic resonance~\cite{du06}. More recently, a variant of the quantum fingerprinting protocol based on coherent states~\cite{arrazola14} has also been implemented experimentally, surpassing the best known classical protocols~\cite{xu15} and even the classical theoretical limit~\cite{guan16}.
    
    However, equality is just one distance measure, and a very special  one.  Here  we  seek  other  measures  of  distance  for which quantum information can achieve a similar exponential advantage. In addition to the inherent theoretical interest of this question in terms of giving  insight  into  the  expressive power of quantum states, quantum   protocol for more general distance measures could find significantly  broader applications.
    
    One example where quantum fingerprinting has been generalised is an efficient quantum communication protocol of Kumar \emph{et al.} based on coherent states~\cite{kumar17}, which can approximately compute the Euclidean distance between unit vectors up to low additive error. This protocol is directly based on the use of the swap test to approximate $\ell_2$-distances between quantum states~\cite{PhysRevLett.87.167902}. There are many other distance measures of practical relevance where it is less clear whether similar techniques to quantum fingerprinting can be applied.
    
    \subsection{Our results}
    
    Our main result is a quantum protocol for approximately computing another distance measure, the Hamming distance, up to low {\em relative} error. This notion of accuracy is important when one wishes to compare objects that are similar; for example, when one of the objects is produced by a small number of errors affecting the other~\cite{cormode00}. Approximating the Hamming distance between two $n$-bit strings up to additive accuracy $\epsilon n$ (analogous to the accuracy achieved by the protocol of~\cite{kumar17}) would give no useful information in this situation.
    
    In the setting we consider, Alice and Bob are given $x, y\in\{0,1\}^n$, respectively. Their goal is to approximately calculate the Hamming distance $d(x,y)$ between $x$ and $y$, i.e., they must output $d_\epsilon(x,y)$ such that $(1-\epsilon)d(x,y) \leq d_\epsilon(x,y) \leq (1+\epsilon)d(x,y)$. Pang and El Gamal~\cite{doi:10.1137/0215065} proved a lower bound of $\Omega(n)$ for exactly calculating the Hamming distance  in the multi-round two-party classical communication model. Here we describe a quantum protocol that approximately computes the Hamming distance in the SMP model by communicating $\poly(\log{n})$ qubits.
    
    \begin{restatable}{thm}{1}
    \label{thm:main}
    There is a quantum protocol in the {\rm SMP} model with private randomness which communicates $O((\log n)^2 (\log \log n) / \epsilon^3)$ qubits and computes the Hamming distance between $n$-bit strings up to relative error $\epsilon$, for any $\epsilon = \Omega(1/\log n)$, with failure probability bounded above by an arbitrarily small constant. 
    \end{restatable}
    
    The protocol is based on a subroutine which determines whether, for some threshold $\delta$, $d(x,y) \le \delta$ or $d(x,y) \ge (1+\epsilon)\delta$. This subroutine maps $x$ and $y$ to $N$-bit strings $Ax$, $Ay$ such that in the first case, $d(Ax,Ay)$ is low (less than $\alpha N$, for some constant $\alpha$), whereas in the second case, $d(Ax,Ay)$ is high (greater than $\beta N$, for some constant $\beta > \alpha$). Alice and Bob then encode the strings $Ax$ and $Ay$ as quantum superpositions, which the referee can distinguish using the swap test~\cite{PhysRevLett.87.167902}.
    
    Note that there exists a corresponding classical protocol in the SMP model with shared randomness, with a similar complexity. One way to see this is that the quantum protocol is ultimately based on the use of the swap test to approximately compute the inner product between unit vectors, for which there is an efficient classical protocol in this model~\cite{kremer99}.
    
    We then generalise Theorem~\ref{thm:main} to other distance measures: in particular, those which can be interpreted as distances in graphs. A graph $G = (V,E)$ is fixed in advance, and each of Alice and Bob is given a vertex of $G$ ($v$ and $w$, respectively). They aim to approximately compute $d_G(v,w)$, the length of a shortest path in $G$ between $v$ and $w$, up to relative error $\epsilon$.
    
    We first observe that Theorem~\ref{thm:main} can be applied to give an efficient protocol for this problem whenever there is a distance-preserving embedding of $G$ into the hypercube: the graph whose vertex set is $\{0,1\}^m$, for some $m$, and where two vertices are connected by an edge whenever their Hamming distance is 1. In fact, this can be generalised further, to graphs which are embeddable into the hypercube such that distances are preserved up to a constant factor $k$. Such graphs are known as $\ell_1$-graphs, because it turns out that this criterion is equivalent to the existence of a distance-preserving embedding of the graph in $\ell_1$~\cite{Chepoi2008}. The class of $\ell_1$-graphs includes all trees, median graphs, Hamming graphs, and Johnson graphs~\cite{Chepoi2008}.
    
    Distances in $\ell_1$-graphs are used in a variety of applications, a few of which we outline here. Partial cubes ($\ell_1$-graphs with embedding constant $k=1$) were initially introduced by Graham and Pollak~\cite{graham1971addressing} as a model for interconnection networks in the Bell System, with distances between vertices corresponding to the number of hops between `loops' in their network. Antimatroids (a specific subclass of $\ell_1$-graphs) are used as structures to represent the required steps to develop a student's knowledge in a certain topic, and the distance between two points that represent concepts in these structures corresponds to the length of a student's learning path~\cite{falmagne2006assessment}. The Barnes-Hut tree method in many-body physics~\cite{Barnes1986} provides a systematic way of determining the degree of `closeness' between two different particles. The distance between two nodes in the tree is linked to this `closeness' property and can be used for various purposes, e.g.\ to calculate gravitational forces in star clusters and study galaxy evolution~\cite{pfalzner2005many}. Tree structures are also used in biology, where phylogenetic trees classify organisms based on overall similarity, and the distance between vertices is related to genetic or mutation distance~\cite{fitch1967construction}.
    
    Our protocol is efficient for $\ell_1$-graphs $G$ whose diameter $\text{diam}(G)$ is low, where the diameter is defined as $\text{diam}(G) := \max_{v,w} d_G(v,w)$.
    \begin{restatable}{thm}{2}
    \label{thm:l1graph}
        Let $G=(V,E)$ be an $\ell_1$-graph with $|V|$ vertices, and let $v,w \in V$. There is a quantum protocol in the {\rm SMP} model with private randomness which communicates $O((\log\operatorname{diam}(G))(\log\log\operatorname{diam}(G))(\log\log|V|)/\epsilon^3)$ qubits and computes $d_G(v,w)$ up to relative error $\epsilon$, for any $\epsilon = \Omega(1/\log \operatorname{diam}(G))$, with failure probability bounded above by an arbitrarily small constant. 
    \end{restatable}
    
    For any graph $G$, even testing equality between vertices requires $\Omega(\sqrt{\log |V|})$ bits of classical communication in the SMP model without shared randomness~\cite{MR1427554}, so this is an exponential separation for those $\ell_1$-graphs where $\text{diam}(G) = O(\log |V|)$, e.g.\ expander graphs like basis graphs of matroids~\cite{anari2019log}. $d_G(v,w)$ can be computed trivially using $O(\log |V|)$ bits of classical communication, by sending the labels of $v$ and $w$ to the referee. So for graphs $G$ where $\operatorname{diam}(G)$ is close to $|V|$, Theorem~\ref{thm:l1graph} gives little or no improvement on the classical complexity. One may wonder whether this is simply a limitation of our protocol, but we show that this is not the case.
    \begin{restatable}{thm}{3}
    \label{thm:lb}
    	Given a graph $G$ with diameter $\operatorname{diam}(G)$, any one-way quantum communication protocol that computes $d_G(v,w)$ up to relative error $\epsilon < 1/4$ with failure probability $1/3$ must transmit at least $\Omega(\log \operatorname{diam}(G))$ qubits.
    \end{restatable}
    
    As every protocol in the SMP model implies a one-way protocol, this shows that the complexity of our protocol is nearly optimal in terms of its dependence on $\operatorname{diam}(G)$.
    
    Finally, we show that our protocol for approximately computing the Hamming distance can be used to give an efficient protocol for approximately computing the $\ell_1$-distance between vectors in $\mathbb{R}^n$.
    
        \begin{restatable}{thm}{4}
    \label{thm:l1vectors}
    Let $x,y \in [-1,1]^n$ such that each entry of $x$ and $y$ is specified by a $k$-bit string, with $k = O(\log n)$. There is a quantum protocol in the {\rm SMP} model which communicates $O((\log n)^2 (\log \log n) / \epsilon^3)$ qubits and computes $\|x-y\|_1$ up to relative error $\epsilon$, for any $\epsilon = \Omega(1/\log n)$, with failure probability bounded above by an arbitrarily small constant. 
    \end{restatable}
    
    A natural special case of Theorem \ref{thm:l1vectors} is when $x$ and~$y$ are probability distributions. Then our result enables Alice and Bob to determine the distance between two distributions, one of which is a small perturbation of the other.
    
    Two interesting questions which remain open are whether one can find a similar result to Theorem \ref{thm:l1graph} which holds for all graphs, without the restriction to $\ell_1$-graphs, and if the communication complexity dependence on $\epsilon$, currently at $1/\epsilon^3$, can be improved.
    
    \subsection{Related work}
    
    The Hamming distance is a fundamental distance measure and has been studied in various forms. In the context of quantum communication complexity, Liu and Zhang~\cite{liu13} gave a quantum sketching protocol for the related ``threshold'' problem of determining whether the Hamming distance is larger than $d$, for some $d$. Their protocol uses $O(d \log n)$ communication, improving a previous $O(d \log^2 n)$ protocol of Gavinsky, Kempe and de Wolf~\cite{gavinsky04}. Huang \emph{et al.}~\cite{huang06} had previously proven an $\Omega(d)$ lower bound for even the two-way quantum communication complexity of the threshold Hamming distance problem, together with an $O(d \log d)$ upper bound in the classical SMP model with public randomness.
    
    A  key  ingredient  in  the  upper  bound  of Huang \emph{et al.} is a protocol which communicates $O(1)$ bits and distinguishes between the case that the Hamming distance is at most $d$, and the case that the Hamming distance is at least $2d$, for arbitrary $d$. Their protocol can be seen as a variant of our Lemma~\ref{lem:lem2} below with $N=1$; similar analysis shows that it could be generalised to distinguish between Hamming distance $d$ and Hamming distance $(1+\epsilon)d$ with $O(1/\epsilon^2)$ bits of communication. Using a generic construction of Yao~\cite{Yao:2003:PQF:780542.780554}, improved by Gavinsky, Kempe, and de Wolf~\cite{gavinsky2006strengths}, this implies a quantum sketching protocol for the same task which communicates $2^{O(1/\epsilon^2)} \log n$ qubits. Using a similar approach to our work, this in turn implies a protocol which solves the approximate Hamming distance problem by transmitting $2^{O(1/\epsilon^2)} \poly \log n$ qubits. This is the same asymptotic complexity as our protocol for constant $\epsilon$, but in practice the $2^{O(1/\epsilon^2)}$ factor makes the protocol infeasible for even modest values of $\epsilon$.
    
    Classically, there has also been substantial work on approximately computing the Hamming distance between a small ``pattern'' and a larger string, both locally and in a distributed context (see~\cite{clifford16} and references therein).

	More generally, the field of communication complexity studies the amount of communication needed between two or more parties to solve a particular problem~\cite{kushilevitz97,buhrman10}. We now give a brief summary of this area. The simplest and most illustrative scenario is the one in which two parties, called Alice and Bob, each possesses some piece of information, often encoded into some string, so that Alice has $x\in X$ and Bob has $y\in Y$, and they want to compute some function $f(x,y)$. Since each does not know the piece of information the other has, they will need to communicate information in order to compute $f(x,y)$. The most straightforward way to solve the problem is to have Alice and Bob exchange their entire string, but sometimes more efficient protocols exist. This communication model was first introduced by Yao in 1979~\cite{Yao:1979:CQR:800135.804414}. 
    
    An important variant of this usual general communication scenario is the model of quantum communication complexity, again introduced by Yao~\cite{yao93}, where now Alice and Bob each has a quantum computer and they exchange qubits instead of bits and/or make use of shared entanglement. The question is whether Alice and Bob can now compute $f$ with less communication than in the classical case; in some cases, this is known to be possible~\cite{buhrman10}.
    
    The above communication scenarios can be narrowed down by imposing some restrictions on the communication process between Alice and Bob, and by restricting or allowing resources like randomness and entanglement. The three most common communication models are the one-way, the multi-round two-party and the simultaneous message passing (SMP) models. In the multi-round two-party model both Alice and Bob can communicate with the other. On the other hand, in the one-way model only one party can communicate with the other, e.g.\ Alice communicates with Bob. Finally, in the SMP model Alice and Bob are only allowed to send messages to a third party, called the referee, who then computes $f(x,y)$. The SMP model was also introduced by Yao (1979)~\cite{Yao:1979:CQR:800135.804414} and is the weakest reasonable model of communication complexity. Considering the SMP model in particular, as previously mentioned, Buhrman \emph{et al.}~\cite{PhysRevLett.87.167902} proved that, if $f$ is the equality function, then a communication reduction from $\Theta(\sqrt{n})$ bits to $\Theta(\log n)$ qubits is possible.
    
    
    
    Later, Yao showed that any classical SMP protocol with shared randomness that transmits $O(1)$ bits and computes a function on $n$ bits implies a quantum SMP protocol without shared randomness that transmits $O(\log n)$ qubits~\cite{Yao:2003:PQF:780542.780554}. This result was generalised by Gavinsky \emph{et al.}~\cite{gavinsky2006strengths}, who gave a quantum SMP protocol that simulates any 2-way quantum communication protocol with shared entanglement, at communication cost exponential in the cost of the original protocol. However, Gavinsky \emph{et al.}\ also proved that for most functions, quantum fingerprinting protocols, which are a subclass of quantum SMP protocols, are exponentially worse than classical deterministic SMP protocols.
    
    Recently, more exotic communication models based on indefinite causal structures were used to demonstrate exponential quantum advantage. Wei \emph{et al.}~\cite{PhysRevLett.122.120504} and Gu\'erin \emph{et al.}~\cite{PhysRevLett.117.100502} showed such an exponential communication advantage by using the concept of a quantum switch (a device that controls the order in which two transformations are performed) to coherently superpose the one-way communication path of information in a tripartite setting, i.e., from Alice to Bob and then to the referee or from Bob to Alice and then to the referee.

	\section{The Protocol}
	
	In this section we present our protocol for approximating the Hamming distance $d(x,y)$ between two strings $x,y\in\{0,1\}^n$ up to relative error $\epsilon$ in the SMP model. That is, 
    Alice and Bob seek the referee to output $d_\epsilon(x,y)$ such that $(1-\epsilon)d(x,y) \leq d_\epsilon(x,y) \leq (1+\epsilon)d(x,y)$. Call this problem $\text{HAM}_\epsilon$.
	
	We first state a lemma that is going to be useful for our protocol.
	\begin{definition} 
		Given an $N$-bit string $x$, define the quantum state
	    \begin{align*}
		    |h_x\rangle = \frac{1}{\sqrt{N}}\sum_{i=1}^N|i\rangle|x_i\rangle,
	    \end{align*}
	    where $x_i$ is the $i$-th bit of $x$.
	\end{definition}
	\begin{lemma}
		Given the $N$-bit strings $x$ and $y$, their Hamming distance $d(x,y)$ can be estimated up to additive accuracy $N\epsilon$ with failure probability $\delta$ using $O(\log(1/\delta)/\epsilon^2)$ copies of $|h_x\rangle$ and $|h_y\rangle$. 
		\label{lem:lem1}
	\end{lemma} 
	\begin{proof}
	    First note that
    	\begin{align*}
	    	\langle h_y|h_x\rangle = \frac{1}{N}\sum_{i=1}^N \langle y_i|x_i\rangle = 1 - \frac{d(x,y)}{N}.
	    \end{align*}
	    The swap test~\cite{PhysRevLett.87.167902} is a test which outputs either $0$ or $1$ on input $\ket{h_x}\ket{h_y}$, and outputs $1$ with probability
	    \begin{align*}
		    p := \frac{1}{2}\left(1 - |\langle h_y|h_x\rangle|^2\right) = \frac{d(x,y)}{N} - \frac{d(x,y)^2}{2N^2}.
	    \end{align*}
	    Before presenting our protocol, we first analyse two different cases: $d(x,y)/N \leq 9/10$ and $d(x,y)/N > 9/10$.
	    
	    Suppose that $d(x,y)/N$ is bounded away from $1$, say $d(x,y)/N \leq 9/10$. This means that $p$ is bounded away from $1/2$. We apply the swap test to $k$ copies of $\ket{h_x}\ket{h_y}$, for some $k$ to be determined. Let $P_i$ correspond to the outcome of the $i$-th swap test and $P:=\frac{1}{k}\sum_{i=1}^kP_i$. By a Chernoff bound~\cite[Theorem~1.1]{dubhashi09}, 
	    \begin{align*}
		    \text{Pr}[|P - p| \geq \epsilon] \leq 2e^{-2k\epsilon^2},
	    \end{align*}
	    which means that we obtain an approximation $P = p\pm \epsilon$ with probability $1-\delta$ by taking $k=O(\log(1/\delta)/\epsilon^2)$. Since $d(x,y)/N = 1 - \sqrt{1-2p}$, we set $\widetilde{d}/N = 1-\sqrt{1-2P}$ as our approximation to $d(x,y)/N$. The derivative of $\sqrt{1-2z}$ is $1/\sqrt{1-2z}$, which is $O(1)$ around $z=p$, since $p$ is bounded away from $1/2$. Then, by a Taylor expansion around $p$, we have that, with probability $1-\delta$,
	    \begin{align*}
	        \left|\frac{\widetilde{d}}{N} - \frac{d(x,y)}{N}\right| = O(|P-p|) = O(\epsilon).
	    \end{align*}
	    
	    On the other hand, if $d(x,y)/N$ is close to $1$, say $d(x,y)/N > 9/10$, then we use the following result due to Yao~\cite[Lemma~1]{Yao:2003:PQF:780542.780554} to bound the outcome of the swap tests,
    	\begin{align*}
    		\text{Pr}\left[\big|\widetilde{d} - d(x,y)\big| \geq N\epsilon \right] \leq 2e^{-k\epsilon^4/32},
    	\end{align*}
    	where $\widetilde{d}/N = 1 - \sqrt{1 - 2P}$, with again $P:=\frac{1}{k}\sum_{i=1}^kP_i$. Hence, if $k=O(\log(1/\delta)/\epsilon^2)$, with probability $1-\delta$ we obtain an approximation $\widetilde{d}/N = d/N \pm \sqrt{\epsilon}$.
	    
	    We now present our protocol for estimating $d(x,y)/N$: fix $k = O(\log(1/\delta)/\epsilon^2)$. For half of the $k$ copies we apply the usual swap test. If the resulting estimate $\widetilde{d}/N$ is at most $2/3$, then we output it as the final answer. Otherwise, by using the other half of the $k$ copies, we apply the swap tests on $\mathbb{I}\otimes X|h_x\rangle$ and $|h_y\rangle$, where $X$ is the usual Pauli operator, and obtain an estimate $\widetilde{d}/N$ (Alice or the referee can apply the operator $\mathbb{I}\otimes X$ to the copies of $|h_x\rangle$). We output $1-\widetilde{d}/N$ as the final answer.
	    
	    To see why this works, first note that, if $d(x,y)/N > 9/10$, then its approximation from the swap tests is such that $\widetilde{d}/N > 9/10 - \sqrt{\epsilon} > 2/3$ for sufficiently small $\epsilon$ and probability $1-\delta/2$. This means that, if the estimate from the first half of the copies is $\leq 2/3$, then it must be an approximation to some $d(x,y)/N \leq 9/10$, and we know that it is within distance $\epsilon$ with probability $1-\delta/2$. 
	    
	    On the other hand, if the estimate of the first half of the copies is $>2/3$, then $d(x,y)/N \geq 2/3 - \epsilon > 1/2$ for sufficiently small $\epsilon$ and with probability $1-\delta/2$. Note that $\langle h_y|\mathbb{I}\otimes X|h_x\rangle = d(x,y)/N$, i.e., estimating the Hamming distance $d/N$ with $\mathbb{I}\otimes X|h_x\rangle$ and $|h_y\rangle$ is equivalent to estimating the Hamming distance $1-d/N$ with $|h_x\rangle$ and $|h_y\rangle$. This means that the swap tests on the second half of the copies will estimate the Hamming distance $1-d(x,y)/N \leq 1/2$, for which an approximation $\widetilde{d}/N$ within distance $\epsilon$ can be obtained with probability $1-\delta/2$. Thus, with overall probability $1-\delta$, our protocol outputs $\widetilde{d}$ such that $|\widetilde{d} - d| \leq \epsilon N$ in both cases.
	\end{proof}
    
    A similar protocol that guarantees the $\epsilon^{-2}$ complexity (and communicated to us by Ronald de Wolf) is to artificially lower the probability of measuring $1$: the referee attaches $2$ extra qubits $|+\rangle|+\rangle$ to each of the states $|h_x\rangle$ and $|h_y\rangle$ and, conditioned on the qubits being $|11\rangle$, applies $\mathbb{I}\otimes X$ to $|h_x\rangle$ and $|h_y\rangle$ separately, i.e., flips the encoded strings. The end result is to enlarge the code's size from $N$ to $4N$ and to change the Hamming distance from $d$ to $3d+(N-d) = N+2d$, which is bounded away from $4N$ for $0\leq d \leq N$.
    
    Given that we aim to approximately compute the inner product between $\ket{h_x}$ and $\ket{h_y}$ in Lemma~\ref{lem:lem1}, the reader may wonder why the Hadamard test~\cite{aharonov09a} was not used instead, since it allows direct estimation of $\langle h_y|h_x\rangle$. The reason is that the Hadamard test requires the ability to produce the coherent superposition $\frac{1}{\sqrt{2}}(\ket{0}\ket{h_x} + \ket{1}\ket{h_y})$, which is not available to the referee.
	
	In the following, we use the notation $|z|$ to mean the number of entries equal to $1$ in a string $z\in\{0,1\}^n$.	
    \begin{lemma}
		Consider an $N\times n$ matrix $A$ over $\mathbb{F}_2$ whose entries are randomly chosen from $\{0,1\}$, and equal to $1$ with independent probability $1/(2d)$ for some $d \ge 1$. Fix $\epsilon > 0$. Then there exist values $\delta_1 < \delta_2$ that do not depend on $N$ and $n$, such that $\delta_2 - \delta_1 = \Theta(\epsilon)$ and for any $\eta > 0$:
		\begin{itemize}
		\item for all $z\in\{0,1\}^n$ such that $|z| \leq d$, $\Pr_A\big[|Az| \geq N \delta_1 + N\eta\big] \leq e^{-2N\eta^2}$;
		\item for all $z\in\{0,1\}^n$ such that $|z| \geq (1 + \epsilon)d$, $\Pr_A\big[|Az| \leq N \delta_2 - N\eta\big] \leq e^{-2N\eta^2}$.
		\end{itemize}
		Hence, for sufficiently large $N = \Theta(n/\epsilon^2)$, with high probability over the choice of $A$, it is sufficient to determine $|Az|$ up to additive accuracy $\Theta(N\epsilon)$ to distinguish between the cases $|z| \le d$ and $|z| \ge (1+\epsilon)d$.
		\label{lem:lem2}
	\end{lemma}
	\begin{proof}
		It is shown in~\cite[Lemma~2.1]{doi:10.1137/S0097539798347177} that for any $z\in\{0,1\}^n$ and $i\in[N]$, $\text{Pr}[(Az)_i=1] = \frac{1}{2}\left(1-(1-1/d)^{|z|}\right)$ and that the probabilities of this event for $|z| \leq d$ and $|z|\geq (1+\epsilon)d$ are bounded by values $\delta_1, \delta_2$ that do not depend on $N$ and $n$ and are separated by $\Theta(1-e^{-\epsilon}) = \Theta(\epsilon)$. That is,
	\begin{subequations}
	\begin{align*}
			\text{Pr}_A\left[(Az)_i=1\right] \leq \delta_1 = \frac{1}{2}\left(1-\left(1-\frac{1}{d}\right)^{d}\right) ~\text{if}~ |z| \leq d, \\
		\begin{multlined}[b][0.46\textwidth]
			\text{Pr}_A\left[(Az)_i=1\right] \geq \delta_2 = \frac{1}{2}\left(1-\left(1-\frac{1}{d}\right)^{(1+\epsilon) d}\right) \\
			~\text{if}~ |z|\geq (1+\epsilon)d.
		\end{multlined}
	\end{align*}
	\end{subequations}
	The expected value of $|Az| = \sum_i (Az)_i$ then satisfies
	\begin{align*}
		&\mathbb{E}[|Az|] \leq N\delta_1 ~\text{if}~ |z|\leq d, \\
		&\mathbb{E}[|Az|] \geq N\delta_2 ~\text{if}~ |z|\geq (1+\epsilon)d.
	\end{align*}
	\par If $|z| \leq d$ so that $\mathbb{E}[|Az|]\leq N\delta_1$, by a Chernoff bound,
	\begin{align*}
		\text{Pr}_A\big[|Az| \geq N\delta_1 + N\eta\big] \leq e^{-2N\eta^2}.
	\end{align*}
	By the same token, if $|z| \geq (1+\epsilon)d$, so that $\mathbb{E}|Az|] \geq N\delta_2$,
	\begin{align*}
		\text{Pr}_A\big[|Az| \leq N\delta_2 - N\eta\big] \leq e^{-2N\eta^2}.
	\end{align*}
	Taking a union bound over all $z\in\{0,1\}^n$ in both cases, we have
	\begin{align*}
		&\begin{multlined}[b][0.48\textwidth]
			\text{Pr}_A\big[\exists z \text{ s.t. } |z|\leq d \text{ and } |Az|\geq N\delta_1 + N\eta\big] \nonumber\\
			\leq 2^ne^{-2N\eta^2} = e^{n\ln{2} - 2N\eta^2}, 
			\end{multlined}\\
		&\begin{multlined}[b][0.48\textwidth]
			\text{Pr}_A\big[\exists z \text{ s.t. } |z|\geq (1+\epsilon) d \text{ and } |Az|\leq N\delta_2 - N\eta\big] \\
			\leq 2^ne^{-2N\eta^2} = e^{n\ln{2} - 2N\eta^2},
		\end{multlined}
	\end{align*}
	so that it is sufficient to choose $N = \Omega(n/\eta^2)$ to bound the probability that either case occurs by an arbitrarily small constant. Choosing $\eta = c\epsilon$ for a sufficiently small constant $c$, we have $|Az| \le N (\delta_1 + c \epsilon)$ if $|z| \le d$, and $|Az| \ge N (\delta_2 - c \epsilon)$ if $|z| \ge (1+\epsilon)d$. Therefore, it is sufficient to determine $|Az|$ up to additive accuracy $O(N\epsilon)$ to distinguish these two cases.
	\end{proof}
	
	The map $A$ in Lemma~\ref{lem:lem2} can be interpreted as a linear code. Such codes are also used in quantum fingerprinting protocols~\cite{PhysRevLett.87.167902,xu15}, but here, unlike previous protocols, we choose the matrix $A$ to be sparse and random. This enables us to control its behaviour when acting on strings $z$ such that $|z| \approx d$ for small $d$.
	
	We now describe our protocol based on the two previous lemmas. In this protocol, Alice and Bob have already agreed beforehand on the matrix $A$, guaranteed to exist by Lemma~\ref{lem:lem2}, to be used. We stress that this matrix is fixed in advance and does not need to be chosen using shared randomness.
	\begin{example}	
    \label{prot:cube}
	Consider the following subroutine for arbitrary $d \in [1,n]$ and $\delta > 0$:
    Alice and Bob encode their $n$-bit strings $x$ and $y$ as $Ax$ and $Ay$, respectively, where $A$ is picked according to Lemma~{\rm \ref{lem:lem2}} and multiplication is over $\mathbb{F}_2$. They send $O(\log (1/\delta)/\epsilon^2)$ copies of the quantum states $|h_{Ax}\rangle$ and $|h_{Ay}\rangle$ to the referee, who performs swap tests and estimates the Hamming distance $d(Ax,Ay)$ up to accuracy $N\epsilon$ with failure probability $\delta$. By Lemma~{\rm \ref{lem:lem2}}, this is sufficient to determine whether $d(x,y) \leq d$ or $d(x,y) \geq (1+\epsilon)d$ with failure probability $\delta$.
	
	Alice and Bob then apply this subroutine to the sequence $S$ of values $d$
	\begin{align*}
		0, 1,1+\epsilon, (1+\epsilon)^2, \dots
	\end{align*}
	where the last element in $S$ corresponds to the minimal $k$ such that $(1+\epsilon)^{k+1} > n$; there are $O(\log n/\log(1+\epsilon)) = O((\log n)/\epsilon)$ elements in the sequence. (In the case $d=0$, they use the standard quantum fingerprinting protocol instead.) Given the $O((\log n)/\epsilon)$ results, the referee outputs the minimal $\widetilde{d}$ such that the subroutine returned ``$d(x,y) \leq \widetilde{d}$''.
    \end{example}
    We first show that, if each use of the subroutine succeeds, the overall algorithm achieves the required level of accuracy. By the definition of $S$, there exist consecutive elements $d_0, d_1, d_2 \in S$ such that $d_0 \le d(x,y)/(1+\epsilon)$, $d(x,y)/(1+\epsilon) \le d_1 \le d(x,y)$, $d(x,y) \le d_2 \le (1+\epsilon)d(x,y)$. Then on input $d_2$ the subroutine must return ``$d(x,y) \leq d_2$'', while for input $d_0$ it must return ``$d(x,y) \geq (1+\epsilon)d_0$'', so the output $\widetilde{d}$ is either $d_1$ or $d_2$ and hence
    \[ (1-\epsilon)d(x,y) \le \frac{d(x,y)}{1+\epsilon} \le \widetilde{d} \le (1+\epsilon)d(x,y). \]
    %
	\par Setting $\delta = O(\epsilon/\log n)$ and using a union bound over the $O((\log n)/\epsilon)$ uses of the subroutine, the probability that any of the subroutines fails can be upper-bounded by an arbitrarily small positive constant.

    Assuming that $\epsilon \geq 1/\log{n}$, the overall communication complexity is
	\begin{align*}
		&O\left(((\log{n})/\epsilon)\cdot(\log{(1/\delta)}/\epsilon^2)\cdot(\log{(n/\epsilon)})\right) = \nonumber\\
		&O\left((\log{n})^2(\log\log{n})/\epsilon^3\right).
	\end{align*}
	This completes the proof of Theorem \ref{thm:main}.
	
	\section{Measuring Distances in Graphs}
	
    In the following, for an arbitrary graph $G$ and vertices $v$, $w$, let $d_G(v,w)$ denote the distance between $v$ and $w$ in $G$, i.e., the length of a shortest path between $v$ and $w$. Also, the hypercube graph $Q_n$ is defined as the graph with vertex set $\{0,1\}^n$, where distance between vertices is the Hamming distance.
	
	The algorithm from last section for approximately measuring the Hamming distance between two strings in the SMP model can be slightly modified to approximately compute the distance between two vertices in specific graphs in the SMP model. That is, to solve the following problem: for some graph $G = (V,E)$, and given vertices $v$, $w$ as input, output $\widetilde{d}$ such that $(1-\epsilon)d_G(v,w) \le \widetilde{d} \le (1+\epsilon)d_G(v,w)$. Call this problem $\text{DIS}_\epsilon[G]$.
    The idea is to embed a given graph $G$ into a hypercube graph such that all the distances between vertices are preserved or rescaled by a constant factor. Once this embedding is achieved, the hypercube structure allows the equivalence between vertex distance in the graph and Hamming distance, so that a binary string can be associated with each vertex and Protocol~\ref{prot:cube} can be applied to these binary strings.
	
	The downside of the above approach is that it cannot be applied to any given graph, since most graphs are not isometrically embeddable into a hypercube. The graphs which can be isometrically embedded into hypercubes are known as partial cubes~\cite{Bonnington2003OnCA, OVCHINNIKOV20085597}. 
	
    The identification of which graphs are partial cubes is an interesting question by itself. The class of partial cubes is relatively broad. The most important examples are hypercubes, trees~\cite{WU1985238} and median graphs~\cite{Ovchinnikov2011}. It also includes other significant classes, e.g.\ tope graphs of oriented matroids (specially graphs of regions of hyperplane arrangements)~\cite{bjorner1999oriented,eppstein2007media}, conditional oriented matroids~\cite{bandelt2018coms,chepoi2020hypercellular}, antimatroids~\cite{KEMPNER2013169,eppstein2007media}, weak orderings~\cite{eppstein2007media}, bipartite $(6, 3)$-graphs~\cite{Chepoi2008}, tiled partial cubes~\cite{Imrich2002}, netlike partial cubes~\cite{POLAT20072704} and two-dimensional partial cubes~\cite{chepoi2020two}.
	
	Partial cubes can be fully characterized via \emph{Djokovi\'c's Characterization}~\cite{djokovic1973distance, deza2009geometry}, introduced by Djokovi\'c in 1973. It connects the property of isometric embedding to bipartiteness and convexity of some specific sub-graphs of the original graph. Here a set is said to be convex if it is closed under taking shortest paths, i.e., if the shortest paths between any two points from the set are also contained in the set. Djokovi\'c's Characterization states, more specifically, that a connected graph $G$ can be isometrically embedded into a hypercube if and only if $G$ is bipartite and $G(a|b)$ is convex for each edge $(a,b)$ of $G$, where $G(a|b) := \{x\in V(G)| \text{ }d_G(x,a) < d_G(x,b)\}$ is the set of the vertices closer to $a$ than $b$. In other words, to check if a graph is a partial cube, one needs to check first if the graph is bipartite. Apart from that, one chooses an edge and constructs the set of all vertices that are closer to one vertex of the chosen edge than the other vertex. Then one needs to check if all shortest paths connecting any two vertices from this set only pass through the vertices of the set. If yes, the set is said to be convex and the same procedure is repeated for another edge of the original graph. If all sets constructed in this way are convex, then the graph is a partial cube.
   	
    Since our protocol is unaffected if all distances are rescaled by a constant factor, the idea of partial cubes can be expanded by the following definitions.
        
	\begin{definition}[\cite{Chepoi2008,Shpectorov:1993:SEG:156547.156553}]
		Given two connected and unweighted graphs $G$ and $H$, we write $G \overset{k}{\hookrightarrow} H$ and say that $G$ is a $k$-scale embedding of $H$ if there exists a mapping $\sigma: V(G) \to V(H)$ such that $d_H(\sigma(a),\sigma(b)) = k\cdot d_G(a,b)$ for all nodes $a,b\in V(G)$.
	\end{definition}
	It is clear that partial cubes are just graphs which can be embedded in a hypercube with a $1$-scale embedding. An example of a graph which is not a partial cube, but can be $k$-scale embedded in a hypercube for $k>1$, is a triangle, which embeds into $Q_3$ with $k=2$.
	
    \begin{definition}[\cite{deza2009geometry}]
    	 A graph $G$ is said to be an $\ell_1$-graph if its path metric $d_G$ is $\ell_1$-embeddable, i.e., there is a map $f$ between $V(G)$ and $\mathbb{R}^m$, for some $m$, such that $d_G(v,w) = \|f(v)-f(w)\|_1$.
    \end{definition}
	\begin{restatable}[{\cite[Theorem~8.3]{Chepoi2008}}]{thm}{5}
		A graph $G$ is an $\ell_1$-graph if and only if it admits a scale embedding into a hypercube.
	\end{restatable}
    This means that the graphs we are interested in are $\ell_1$-graphs. This class of $\ell_1$-graphs includes new graphs that are not partial cubes, e.g.\ Hamming graphs, half cubes and Johnson graphs are $2$-embeddable into a hypercube~\cite{Chepoi2008}. As far as we are aware, there is no Djokovi\'c-type characterisation for $\ell_1$-graphs (see e.g.~\cite[Problem~21.4.1]{deza2009geometry}). There has been some progress in characterising the isometric embedding of graphs into Hamming~\cite{chepoi1988isometric} and Johnson~\cite{chepoi2017distance} graphs (hence $2$-embeddability into a hypercube) by considering the convexity property of the sets $G(a|b), G(b|a)$ and $G_{=}(a,b) := \{x\in V(G)|~d_G(x,a) = d_G(x,b)\}$ (and other related sets).
	
	Before stating the communication protocol in the SMP model to approximately measure the distance between two vertices in an $\ell_1$-graph, we state the Johnson-Lindenstrauss lemma~\cite{johnson1984extensions,dasgupta1999elementary,gavinsky2006strengths}, which is going to be useful to reduce the protocol complexity. Note that we use Dirac notation for vectors which are not necessarily normalized.
	\begin{lemma}[Johnson-Lindenstrauss lemma]
	\label{lem:jl}
		Consider $0 < \epsilon < 1/2$ and a positive integer $n$. Then for any set $U$ of $k$ vectors in $\mathbb{R}^n$, there is a linear map $f: \mathbb{R}^n \to \mathbb{R}^{O((\log{k})/\epsilon^2)}$ such that for all $\ket{u},\ket{v}\in U$,
		\[
			(1-\epsilon)\|\ket{u} - \ket{v}\|^2 \leq \| f\ket{u} - f\ket{v}\|^2 \leq (1+\epsilon)\| \ket{u} - \ket{v}\|^2.
		\]
	\end{lemma}
	To find a map $f$ achieving the bounds of Lemma~\ref{lem:jl}, one can choose it at random from an appropriate distribution. A number of different constructions of such random maps are known; one simple example is a suitably normalised projection onto a random subspace of $\mathbb{R}^n$.
	
	As mentioned, e.g.\ in~\cite{gavinsky2006strengths}, if the set $U$ includes the $0$-vector, then the map $f$ also approximately preserves the inner product between all the pairs of vectors in $U$ up to additive error. This implies the following lemma.
    \begin{lemma}
    	Let $0 < \epsilon < 1/2$. Let $U$ be a set of unit vectors in $\mathbb{R}^n$ and let $f:\mathbb{R}^n \to \mathbb{R}^m$ be a linear map such that, for all $\ket{u},\ket{v}\in U\cup\{\vec{0}\}$,
        \[
			(1-\epsilon)\|\ket{u} - \ket{v}\|^2 \leq \| f\ket{u} - f\ket{v}\|^2 \leq (1+\epsilon)\| \ket{u} - \ket{v}\|^2.
		\]
        Define the unit vectors $\ket{\widetilde{u}} = f\ket{u}/\|f\ket{u}\|$ for all $\ket{u}\in U$. Then, for all $\ket{u},\ket{v}\in U$,
        \[
        	\Big| \big|\braket{\widetilde{u}|\widetilde{v}}\big| - \big|\langle u|v\rangle\big|\Big| \leq 4\epsilon.
        \]
    \end{lemma}
    \begin{proof}
    	For clear notation, define $\ket{u'} := f\ket{u}$. By the conditions on $f$, we have that
		\begin{align*}
		1-\epsilon &\leq \langle u' |u'\rangle \leq 1+\epsilon, \\
		(1-\epsilon)\| \ket{u} - \ket{v}\|^2 &\leq \| \ket{u'} - \ket{v'}\|^2 \leq (1+\epsilon)\| \ket{u} - \ket{v}\|^2
		\end{align*}
    for all $\ket{u},\ket{v}\in U$, where the first line was obtained by taking the $0$-vector as one of the vectors and using linearity of $f$. From the above inequalities it follows that
    	\[
        	(1+\epsilon)\langle u|v\rangle -2\epsilon \leq \langle u'|v'\rangle \leq (1-\epsilon)\langle u|v\rangle + 2\epsilon.
        \]
        These new inequalities in turn lead to
        \begin{align*}
        	\langle \widetilde{u}|\widetilde{v}\rangle \geq \frac{(1+\epsilon)\langle u |v\rangle - 2\epsilon}{1+\epsilon} \geq \langle u |v\rangle - 2\epsilon,\\
        	\langle \widetilde{u}|\widetilde{v}\rangle \leq \frac{(1-\epsilon)\langle u |v\rangle + 2\epsilon}{1-\epsilon} \leq \langle u |v\rangle + 4\epsilon,
		\end{align*}
        using that $0 < \epsilon < 1/2$. Therefore
        \[
        	\Big|\big| \langle \widetilde{u}|\widetilde{v}\rangle\big| - \big|\langle u|v\rangle\big|\Big| \leq \Big| \langle \widetilde{u}|\widetilde{v}\rangle - \langle u|v\rangle\Big| \leq 4\epsilon.\qedhere
        \]
    \end{proof}
    \par
    Consider applying Lemma~\ref{lem:lem1} to the normalized quantum states $|\widetilde{h}_x\rangle$ and $|\widetilde{h}_y\rangle$ that are produced by using the Johnson-Lindenstrauss lemma, in the sense that the original states $|h_x\rangle$ and $|h_y\rangle$ in Lemma~\ref{lem:lem1} are replaced with the states $|\widetilde{h}_x\rangle$ and $|\widetilde{h}_y\rangle$, respectively. We argue that this does not change the parameters of the lemma substantially.
    To see that, we note $\big|\widetilde{\eta} - |\langle h_y|h_x\rangle|\big| + \big||\langle \widetilde{h}_y|\widetilde{h}_x\rangle| - |\langle h_y|h_x\rangle|\big| \geq \big|\widetilde{\eta} - |\langle \widetilde{h}_y|\widetilde{h}_x\rangle|\big|$ and hence $\big|\widetilde{\eta} - |\langle \widetilde{h}_y|\widetilde{h}_x\rangle|\big| \geq 5\epsilon \implies \big|\widetilde{\eta} - |\langle h_y|h_x\rangle|\big| \geq \epsilon$, which means
    \begin{align*}
    	\text{Pr}\big[\big|\widetilde{\eta} - |\langle \widetilde{h}_y|\widetilde{h}_x\rangle|\big| \geq 5\epsilon\big] \leq \text{Pr}\big[\big|\widetilde{\eta} - |\langle h_y|h_x\rangle|\big| \geq \epsilon\big],
    \end{align*}
    where $\widetilde{\eta}$ is as defined in Lemma~\ref{lem:lem1}.
    \par With this in mind, and recalling that $\text{diam}(G)$ is defined to be the diameter of the graph $G$, i.e., the greatest distance between any pair of vertices, we present the communication protocol.
	
	\begin{example}
		Alice and Bob each hold vertices $u, v \in V(G)$, respectively, from a graph $G$ which admits a $k$-scale embedding into a hypercube $Q_n$, for some $n$. Their vertex images are the $n$-bit strings $x,y\in Q_n$, respectively. The communication protocol to measure $(1\pm\epsilon)d_G(u,v)$ can be divided into three parts.
		\par First, given $d\in[1,\operatorname{diam}(G)]$ and a matrix $A$ picked according to Lemma~{\rm \ref{lem:lem2}}, Alice and Bob encode their $n$-bit strings $x$ and $y$ as $Ax$ and $Ay$, respectively, where multiplication is over $\mathbb{F}_2$. Differently from the original protocol, Alice and Bob apply the Johnson-Lindenstrauss lemma to their data $Ax$ and $Ay$, which are then encoded into the quantum states $|\widetilde{h}_{Ax}\rangle$ and $|\widetilde{h}_{Ay}\rangle$. There are $|V|$ possible vectors to encode, so the number of used qubits is reduced from $O(\log{n} + \log(1/\epsilon))$ to $O(\log\log{|V|} + \log(1/\epsilon))$.
		\par Second, Alice and Bob send $O(\log (1/\delta)/\epsilon^2)$ copies of their quantum states $|\widetilde{h}_{Ax}\rangle$ and $|\widetilde{h}_{Ay}\rangle$ to the referee, who performs swap tests and estimates the Hamming distance $d(Ax,Ay)$ up to accuracy $N\epsilon$ with failure probability $\delta$, and from this decides if $d(x,y) \leq d$ or $d(x,y) \geq (1+\epsilon)d$.
		\par The third and final part is to apply the first and second parts to the sequence $S$ of values $d$
	\begin{align*}
		0,1, 1+\epsilon, (1+\epsilon)^2, \dots
	\end{align*}
		where the last element in $S$ corresponds to the minimal $k$ such that $(1+\epsilon)^{k+1} > \operatorname{diam}(G)$; there are $O((\log \operatorname{diam}(G))/\epsilon)$ elements in the sequence. Based on the results from the swap tests, the referee outputs $\widetilde{d}$ such that $(1-\epsilon)d(x,y) \leq \widetilde{d} \leq (1+\epsilon)d(x,y)$, in the same way as in Protocol~{\rm \ref{prot:cube}}.
    \end{example}
	\par Setting $\delta = O(\epsilon/\log \text{diam}(G))$, the overall communication complexity is then
	\begin{align*}
		O((\log\text{diam}(G))(\log\log\text{diam}(G))(\log\log|V|)/\epsilon^3),
	\end{align*}
	assuming that $\epsilon \geq 1/(\log{\text{diam}(G)})$.
	
	The performance of the protocol is limited by the diameter of the graph. It is known that $\log_{\Delta-1}|V| - \frac{2}{\Delta} \leq \text{diam}(G) < |V|$, where $\Delta$ is the maximum vertex degree~\cite{chung1987diameters}. If $\text{diam}(G) = O(\log|V|)$, e.g.\ for expander graphs including basis graphs of matroids~\cite{anari2019log}, the overall complexity is polyloglog in $|V|$.  On the other hand, if $\text{diam}(G) = \Theta(|V|)$, the overall complexity is polylog in $|V|$, which is no better than the trivial protocol where Alice and Bob send their entire inputs to the referee.
    
    \subsection{Lower bound}
    
    \par One can ask if there could exist other protocols substantially more efficient than ours. In order to answer this, we prove a lower bound on the quantum communication complexity for the problem of approximately calculating the graph distance between two vertices on a graph, which demonstrates that our protocol is essentially optimal in terms of the dependence of its complexity on the graph diameter. We do not know whether the $3$th-power dependence on $\epsilon$ is optimal, and suspect that it may not be.
    \par The idea behind our proof is to transform the approximate graph distance problem into the problem of approximating the modulus of the difference between two integers. We then show that two uses of a protocol for this approximate modulus problem can compute the Greater-Than function in the one-way communication model. It was shown by Zhang~\cite{zhang11} that the one-way quantum communication complexity of this problem is maximal, improving a previous lower bound of Klauck~\cite{klauck00} by a logarithmic term. The bound of~\cite{zhang11} is used to obtain the lower bound for the approximate modulus problem, and hence for the approximate graph distance problem.
    \par The first step of our proof is to show that two uses of a protocol for the approximate modulus problem can solve the Greater-Than function in the one-way communication model. Consider the Greater-Than problem (GT) defined by the Boolean function $\text{GT}: \{0,1\}^m \times \{0,1\}^m \to \{0,1\}$ as
    \[
    	\text{GT}(x,y) = \begin{cases}
        	1 \text{ if } x \geq y,\\
            0 \text{ if } x < y,
        \end{cases}
    \]
    where $x$ and $y$ are interpreted as $m$-bit integers.
    Given $0 \leq \epsilon < 1$, consider the approximate modulus problem where Alice and Bob are each given an integer $x$ and $y$ (respectively), each expressed as an $m$-bit string, and seek to output $\widetilde{d}$ such that $(1-\epsilon)|x-y| \le \widetilde{d} \le (1+\epsilon)|x-y|$. Call this problem $\text{MOD}_\epsilon$.
    In the following we prove that two uses of this protocol suffice to solve the GT problem.
    Let $\mathcal{P}$ be a quantum communication protocol in the one-way communication model which solves a problem $f$ with failure probability $\delta$. Let $Q^1(\mathcal{P})$ be the communication cost of the protocol $\mathcal{P}$ (in qubits) and denote by $Q^1(f) = \min_{\mathcal{P}}Q^1(\mathcal{P})$ the minimum communication cost over all protocols $\mathcal{P}$ that solve $f$ with failure probability $1/3$.
    \begin{lemma}
    For any $\epsilon < 1/4$,
    	$Q^1(\text{GT}) = O(Q^1(\text{MOD}_\epsilon))$.
        \label{lem:lem10}
    \end{lemma}
    \begin{proof}
    	Let $\mathcal{P}_{MOD}$ be a communication protocol for $\text{MOD}_\epsilon$ in the one-way communication model with failure probability $1/6$. (We can obtain a protocol which achieves this failure probability and communicates $O(Q^1(\text{MOD}_\epsilon))$ qubits using $O(1)$ repetitions of the protocol which achieves failure probability $1/3$ and communicates $Q^1(\text{MOD}_\epsilon)$ qubits.)
        
        Two uses of $\mathcal{P}_{MOD}$ suffice to obtain a communication protocol for $\text{GT}$ in the one-way communication model with failure probability $1/3$ as follows: Alice and Bob apply the protocol $\mathcal{P}_{MOD}$ using $x$ and $y$ as inputs and Bob obtains $z_0\in [(1-\epsilon)|x-y|, (1+\epsilon)|x-y|]$. They both apply the same protocol again, but now Bob inputs $y + z_0$ (Alice still inputs $x$). Bob obtains $z_1$. If $z_0 < z_1$, then $x < y$ and he outputs $0$. Otherwise, $x \ge y$ and he outputs $1$.
        \par To see why this protocol works (assuming that each use of $\mathcal{P}_{MOD}$ succeeds), note that if $x < y$, then $(2-\epsilon)|x-y| \leq |x - y - z_0| \leq (2+\epsilon)|x-y|$, and hence
        \[
        	(2-\epsilon)(1-\epsilon)|x-y| \leq z_1 \leq (2+\epsilon)(1+\epsilon)|x-y|.
        \]
        If $x \geq y$, then $0 \leq |x - y - z_0| \leq \epsilon|x-y|$, and hence
        \[
        	0 \leq z_1 \leq \epsilon(1+\epsilon)|x-y|.
        \]
        For $x < y$ we want to have $z_0 < z_1$, i.e.\ $1+\epsilon < (2-\epsilon)(1-\epsilon)$, which holds if $\epsilon < 2 - \sqrt{3}$. And for $x \geq y$ we need $z_0 \ge z_1$, i.e.\ $\epsilon(1+\epsilon) \le 1 - \epsilon$, which holds if $\epsilon \le \sqrt{2} - 1$. Therefore, by taking $\epsilon < 1/4$, for example, one can distinguish the cases $x < y$ and $x \ge y$ through a comparison between $z_0$ and $z_1$.
        \par Given that every protocol for $\text{MOD}_\epsilon$ in the one-way communication model implies a protocol for $\text{GT}$, we conclude that $Q^1(\text{GT}) = O(Q^1(\text{MOD}_\epsilon))$.
    \end{proof}
    \par The next step is to reduce the approximate graph distance problem to the approximate modulus problem, which we achieve as follows. Let $G$ be a graph with diameter $\text{diam}(G)$. By the definition of diameter, there exists a path graph $P_n \subseteq G$ with $n = \text{diam}(G)$. Therefore, a lower bound for the approximate graph distance problem on $P_n$ implies a lower bound for the same problem on $G$.
    
    The vertices of $P_n$ can be listed in the order $v_1, v_2,..., v_n$ such that the edges are $(v_i, v_{i+1})$, where $i = 1, 2,..., n - 1$. A given vertex $v_i$ can then be labeled by a binary string $x_i\in\{0,1\}^m$, with $m = \Theta(\log n)$, and hence, given $v_i,v_j\in G$, $d_G(v_i,v_j) = |x_i - x_j|$. Therefore, a communication protocol which outputs $\widetilde{d}$ such that $(1-\epsilon)d_G(v_i,v_j) \leq \widetilde{d} \leq (1+\epsilon)d_G(v_i,v_j)$ is equivalent to a communication protocol which solves $\text{MOD}_\epsilon$ on inputs $x_i$, $x_j$. So computing an approximate modulus reduces to computing an approximate graph distance.
    
    With this in mind, we can state our lower bound.
    \begin{restatable}{thm}{6}
    	Given a graph $G$ with diameter $\operatorname{diam}(G)$, the quantum communication complexity for the problem $\operatorname{DIS}_\epsilon[G]$ in the one-way communication model with $\epsilon < 1/4$ and failure probability $1/3$ is $Q^1(\operatorname{DIS}_\epsilon[G]) = \Omega(\log{\operatorname{diam}(G)})$.
    \end{restatable}
    \begin{proof}
    	As mentioned before, the approximate graph distance problem on a path graph $P_n \subseteq G$ with $n = \text{diam}(G)$ should be at least as hard as the same problem on $G$, i.e., $Q^1(\text{DIS}_\epsilon[G]) \geq Q^1(\text{DIS}_\epsilon[P_n])$. Moreover, $\text{DIS}_\epsilon[P_n]$ is equivalent to $\text{MOD}_\epsilon$ on inputs of size $m = \Theta(\log{\text{diam}(G)})$, hence $Q^1(\text{DIS}_\epsilon[G]) \geq Q^1(\text{MOD}_\epsilon)$. According to Lemma~\ref{lem:lem10}, $Q^1(\text{MOD}_\epsilon) = \Omega(Q^1(\text{GT}))$, but $Q^1(\text{GT}) = \Theta(m)$~\cite[Appendix B]{zhang11}, therefore $Q^1(\text{DIS}_\epsilon[G]) = \Omega(\log{\text{diam}(G)})$.
    \end{proof}
    The above result also holds for the SMP model and then states that our communication protocol is optimal in terms of its dependence on $\text{diam}(G)$.

	\section{Measuring $\ell_1$ Distances}
    
    \par As seen in the previous sections, our communication protocol for approximating the Hamming distance can be adapted to $\ell_1$-graphs. A graph $G$ is said to be an $\ell_1$-graph if there exist vectors $u_1,...,u_n \in \mathbb{R}^m$ for some $m$, and with $n = |V(G)|$, such that $d_G(v_i,v_j) = \| u_i - u_j\|_1$ for all $v_i,v_j\in V(G)$. This connection between graphs and $\ell_1$-norm suggests an application of our approximate Hamming distance protocol to $\ell_1$-distances. More specifically, consider the following problem: Alice and Bob are each given a vector $x$, $y$ (respectively) from $[-1,1]^d$. Each entry of each vector is specified by $k$ bits, for some $k$ (1 bit to specify the sign, and $k-1$ bits $z_1,\dots,z_{k-1}$ to specify a binary fraction $z_1 2^{-1} + z_2 2^{-2} + \dots + z_{k-1} 2^{1-k}$). Alice and Bob's task is to approximate the $\ell_1$-distance between $x$ and $y$ up to relative error $\epsilon$ in the SMP model.
    
    A natural special case of this problem is where Alice and Bob are each given a probability distribution $x,y \in \mathbb{R}^d$, respectively, and are asked to approximately compute the $\ell_1$-distance between them (equivalently, the total variation distance, which is defined as half the $\ell_1$-distance). This corresponds to the special case where $x_i,y_i \ge 0$ for all $i$, and $\sum_i x_i = \sum_i y_i = 1$.
    
    Alice and Bob can use our approximate Hamming distance protocol to approximately compute $\|x-y\|_1$: the idea is to map these vectors into a Hamming metric via a unary representation~\cite{Linial1995}.
    Each entry $z \in [-1,1]$ of each vector is mapped to a $2^k$-bit string $s(z)$ such that the first $2^{k-1} (z+1)$ bits of $s(z)$ are set to 1, and the remaining bits are set to 0. Then, for any $z$, $w$, $|z-w| = d(s(z),s(w))/2^{k-1}$. Letting $s(x)$ denote the result of applying this map to each entry of $x$ and concatenating the results, we have $\|x-y\|_1 = d(s(x),s(y))/2^{k-1}$ for bit strings $s(x)$, $s(y)$ of length $m=2^k d$. So we can use our usual communication protocol (Protocol~\ref{prot:cube}) to deliver an estimate of $\|x-y\|_1$ up to relative error $\epsilon$ using $O((\log^2 m)(\log \log m)/\epsilon^3)$ qubits of communication, which is $O((\log^2 d)(\log \log d)/\epsilon^3)$ when $k \le \log d$.
    
    The use of a unary representation may seem wasteful, but a straightforward binary representation would not preserve distances correctly for all inputs. There is also a lower bound that the communication complexity of this problem must have at least a linear dependence on $k$: by the lower bound on the complexity of the MOD$_\epsilon$ problem that follows from Lemma~\ref{lem:lem10}, $\Omega(k)$ bits of communication are required to approximately compute $\|x-y\|_1$ even for $d = 1$. Finally, the protocol can easily be extended to the setting where $x,y \in [-M,M]^d$, for some $M \ge 1$, by rescaling the vectors appropriately.

\section{Conclusions}

We developed an efficient quantum communication protocol to approximately compute the Hamming distance between two $n$-bit strings in the SMP model up to relative error $\epsilon$, which uses $\widetilde{O}((\log{n})^2/\epsilon^3)$ qubits of communication, whereas any classical communication protocol needs to transmit at least $\Omega(\sqrt{n})$ bits. We stress that the protocol approximates the Hamming distance up to a relative, and not additive, error, so that small Hamming distances are approximated accurately.

The Hamming distance protocol was modified to approximate the distance between
any two vertices in a graph. This modification was based on embedding the graph into a subgraph of the Hamming cube such that all distances are preserved up to a constant factor. This requirement restricts the class of graphs to which the original Hamming distance protocol can be applied. Graphs with this property are known as $\ell_1$-graphs. The modified quantum communication protocol to approximate the vertex distance in $\ell_1$-graphs in the SMP model up to relative error $\epsilon$ transmits $\widetilde{O}(\log(\text{diam}(G))/\epsilon^3)$ qubits, where $\text{diam}(G)$ is the diameter of the graph, so the protocol is only efficient for low diameter graphs. A lower bound on the number of communicated qubits needed to approximate the vertex distance shows that this limitation of our protocol is due to the problem itself. More specifically, we proved that any one-way quantum protocol to approximate the distance between any two vertices in a graph needs to communicate at least $\Omega(\log\text{diam}(G))$ qubits. Finally, the original Hamming distance protocol was also modified to approximate the $\ell_1$-norm between two vectors $x,y\in[-1,1]^d$ specified by $k \leq \log{d}$ bits.

Two interesting questions remain open: can one find a similar result to Theorem~\ref{thm:l1graph} which holds for all graphs, without the restriction to $\ell_1$-graphs? And can one improve the communication complexity dependence on $\epsilon$?

No new data were created during this study. 

\subsection*{Acknowledgements}
 
We thank Rapha\"el Clifford for helpful discussions during the course of this work. JFD thanks Noah Linden and Ronald de Wolf for pointing out a small error in Lemma~\ref{lem:lem2}, for helping improving Lemma~\ref{lem:lem1} and for comments on the paper. Finally, we thank Victor Chepoi for spotting a mistake in our previous characterisation of $\ell_1$-graphs and for pointing out many references in this direction. We acknowledge support from the QuantERA ERA-NET Cofund in Quantum Technologies implemented within the European Union's Horizon 2020 Programme (QuantAlgo project). AM was supported by EPSRC Early Career Fellowship Grant No.\ EP/L021005/1. JFD was supported by the Bristol Quantum Engineering Centre for Doctoral Training, EPSRC Grant No.\ EP/L015730/1.

\bibliographystyle{unsrt}
\bibliography{doriguello-montanaro}

\end{document}